\documentclass[runningheads]{llncs}
\usepackage{graphicx}
\usepackage{amsfonts}
\usepackage[font=small,labelfont=small]{caption}
\usepackage{subcaption}
\usepackage{multirow}
\usepackage{amsmath}
\usepackage{enumerate}
\usepackage{multicol}

\newcommand{\E}{\mathbb{E}}
\newcommand{\N}{\mathbb{N}}

\newcommand{\reg}{\textnormal{reg}}
\newcommand{\poly}{\textnormal{poly}}

\newcommand{\veca}{\textbf{\textup{a}}}
\newcommand{\vecp}{\textbf{\textup{p}}}
\newcommand{\otherp}{p}
\newcommand{\matX}{\textbf{\textup{X}}}
\newcommand{\vecu}{\textbf{\textup{u}}}
\newcommand{\G}[3]{\mathcal{G}\left(#1,#2,#3\right)}

\newcommand{\GL}[3]{\mathcal{G}_{#1}\left(#2,#3\right)}
\newcommand{\AS}{\texttt{ApproximateSink}}

\begin{document}
\title{Lower Bounds for the Query Complexity of Equilibria in Lipschitz Games}
\titlerunning{Query Complexity of Equilibria in Lipschitz Games}
%
\author{Paul W. Goldberg\orcidID{0000-0002-5436-7890} \and
Matthew J. Katzman\orcidID{0000-0001-8147-9110}}
\authorrunning{P.\,W. Goldberg and M.\,J. Katzman}
%
\institute{Department of Computer Science, University of Oxford, United Kingdom\\
\email{\{paul.goldberg,matthew.katzman\}@cs.ox.ac.uk}}
\maketitle              
\begin{abstract}
Nearly a decade ago, Azrieli and Shmaya introduced the class of $\lambda$-Lipschitz games in which every player's payoff function is $\lambda$-Lipschitz with respect to the actions of the other players.  They showed that such games admit $\epsilon$-approximate pure Nash equilibria for certain settings of $\epsilon$ and $\lambda$.  They left open, however, the question of how hard it is to find such an equilibrium.  In this work, we develop a query-efficient reduction from more general games to Lipschitz games.  We use this reduction to show a query lower bound for any randomized algorithm finding $\epsilon$-approximate \emph{pure} Nash equilibria of $n$-player, binary-action, $\lambda$-Lipschitz games that is exponential in $\frac{n\lambda}{\epsilon}$.  In addition, we introduce ``Multi-Lipschitz games,'' a generalization involving player-specific Lipschitz values, and provide a reduction from finding equilibria of these games to finding equilibria of Lipschitz games, showing that the value of interest is the sum of the individual Lipschitz parameters.  Finally, we provide an exponential lower bound on the \emph{deterministic} query complexity of finding $\epsilon$-approximate \emph{correlated} equilibria of $n$-player, $m$-action, $\lambda$-Lipschitz games for strong values of $\epsilon$, motivating the consideration of explicitly randomized algorithms in the above results. Our proof is arguably simpler than those previously used to show similar results.

\end{abstract}
\section{Introduction}

A Lipschitz game is a multi-player game in which there is an additive limit $\lambda$ (called the Lipschitz constant of the game) on how much any player's payoffs can change due to a deviation by any other player. Thus, any player's payoff function is $\lambda$-Lipschitz continuous as a function of the other players' mixed strategies. Lipschitz games were introduced about ten years ago by Azrieli and Shmaya \cite{AS13}. A key feature of Lipschitz games is that they are guaranteed to have approximate Nash equilibria \emph{in pure strategies}, where the quality of the approximation depends on the number of players $n$, the number of actions $m$, and the Lipschitz constant $\lambda$. In particular, \cite{AS13} showed that this guarantee holds (keeping the number of actions constant) for Lipschitz constants of size $o(1/\sqrt{n \log n})$ (existence of pure approximate equilibria is trivial for Lipschitz constants of size $o(1/n)$ since then, players have such low effect on each others' payoffs that they can best-respond independently to get a pure approximate equilibrium). The general idea of the existence proof is to take a mixed Nash equilibrium (guaranteed to exist by Nash's theorem \cite{N51}), and prove that there is a positive probability that a pure profile sampled from it will constitute an approximate equilibrium.

As noted in \cite{AS13} (and elsewhere), solutions in pure-strategy profiles are a more plausible and satisfying model of a game's outcome than solutions in mixed-strategy profiles. On the other hand, the existence guarantee raises the question of how to \emph{compute} an approximate equilibrium. In contrast with potential games, in which pure-strategy equilibria can often be found via best- and better-response dynamics, there is no obvious natural approach in the context of Lipschitz games, despite the existence guarantee. The general algorithmic question (of interest in the present paper) is:
\begin{quote}
\emph{Given a Lipschitz game, how hard is it to find a pure-strategy profile that constitutes an approximate equilibrium?}
\end{quote}
Recent work \cite{DFS20,GCW19} has identified algorithms achieving additive constant approximation guarantees, but as noted by Babichenko \cite{Bab19}, the extent to which we can achieve the pure approximate equilibria that are guaranteed by \cite{AS13} (or alternatively, potential lower bounds on query or computational complexity) is unknown.

Variants and special cases of this question include classes of Lipschitz games having a concise representation, as opposed to unrestricted Lipschitz games for which an algorithm has query access to the payoff function (as we consider in this paper). In the latter case, the question subdivides into what we can say about the query complexity, and about the computational complexity (for concisely-represented games the query complexity is low, by Theorem 3.3 of \cite{GR16}). Moreover, if equilibria can be easily computed, does that remain the case if we ask for this to be achievable via some kind of natural-looking decentralized process? Positive results for these questions help us to believe in ``approximate pure Nash equilibrium'' as a solution concept for Lipschitz games. Alternatively, it is of interest to identify computational obstacles to the search for a Nash equilibrium.

\subsection{Prior Work}

In this paper we apply various important lower bounds on the query complexity of computing approximate Nash equilibria of \emph{unrestricted} $n$-player games. In general, lower bounds on the query complexity are known that are exponential in $n$ (which motivates a focus on subclasses of games, such as Lipchitz games, and others).  Hart and Mansour \cite{HM10} showed that the communication (and thus query) complexity of computing an exact Nash equilibrium (pure or mixed) in a game with $n$ players is $2^{\Omega(n)}$.  Subsequent results have iteratively strengthened this lower bound.  First, Babichenko \cite{Bab16} showed an exponential lower bound on the \emph{randomized} query complexity of computing an $\epsilon$-well supported Nash equilibrium (an approximate equilibrium in which every action in the support of a given player's mixed strategy is an $\epsilon$-best response) for a constant value of $\epsilon$, even when considering $\delta$-distributional query complexity, as defined in Definition \ref{def:query}.  Shortly after, Chen, Cheng, and Tang \cite{CCT17} showed a $2^{\Omega\left(n/\log n\right)}$ lower bound on the randomized query complexity of computing an $\epsilon$-approximate Nash equilibrium for a constant value of $\epsilon$, which Rubinstein \cite{Rub16} improved to a $2^{\Omega(n)}$ lower bound, even allowing a constant fraction of players to have regret greater than $\epsilon$ (taking regret as defined in Definition \ref{def:reg}). These intractability results motivate us to consider a restricted class of games (Lipschitz, or large, games) which contain significantly more structure than do general games.

Lipschitz games were initially considered by Azrieli and Shmaya \cite{AS13}, who showed that any $\lambda$-Lipschitz game (as defined in Section \ref{subsec:game}) with $n$ players and $m$ actions admits an $\epsilon$-approximate \emph{pure} Nash equilibrium for any $\epsilon\geq\lambda\sqrt{8n\log2mn}$ \footnote{General Lipschitz games cannot be written down concisely, so we assume black-box access to the payoff function of a Lipschitz game. This emphasizes the importance of considering \emph{query complexity} in this context. Note that a pure approximate equilibrium can still be checked using $mn$ queries.}.  In Section \ref{sec:results} we provide a lower bound on the query complexity of finding such an equilibrium.

Positive algorithmic results have been found for classes of games that combine the Lipschitz property with others, such as \emph{anonymous} games \cite{DP15} and \emph{aggregative} games \cite{Bab18}. For anonymous games (in which each player's payoffs depend only on the \emph{number} of other players playing each action, and not \emph{which} players), Daskalakis and Papadimitriou \cite{DP15} improved upon the upper bound of \cite{AS13} to guarantee the existence of $\epsilon$-approximate pure Nash equilibria for $\epsilon=\Omega(\lambda)$ (the only dependence on $n$ coming from $\lambda$ itself). Peretz et al.\ \cite{PSS20} analyze Lipschitz values that result from $\delta$-perturbing anonymous games, in the sense that every player is assumed to randomize uniformly with probability $\delta$. Goldberg and Turchetta \cite{GT17} showed that a $3\lambda$-approximate pure Nash equilibrium of a $\lambda$-Lipschitz anonymous game can be found querying $O(n\log n)$ individual payoffs.

Goldberg et al.\ \cite{GCW19} showed a logarithmic upper bound on the randomized query complexity of computing $\frac{1}{8}$-approximate Nash equilibria in binary-action $\frac{1}{n}$-Lipschitz games.  They also presented a randomized algorithm finding a $(\frac{3}{4}+\alpha)$-approximate Nash equilibrium when the number of actions is unbounded.

\subsection{Our Contributions}

The primary contribution of this work is the development and application of a query-efficient version of a reduction technique used in \cite{AS13,Bab13a} in which an algorithm finds an equilibrium in one game by reducing it to a population game with a smaller Lipschitz parameter.  As the former is a known hard problem, we prove hardness for the latter.

In Section \ref{sec:prelim} we introduce notation and relevant concepts, and describe the query model assumed for our results.  Section \ref{sec:results} contains our main contributions.  In particular, Theorem \ref{thm:pure} utilizes a query-efficient reduction to a population game with a small Lipschitz parameter while preserving the equilibrium.  Hence, selecting the parameters appropriately, the hardness of finding well-supported equilibria in general games proven in \cite{Bab16} translates to finding approximate pure equilibria in Lipschitz games.  Whilst several papers have discussed both this problem and this technique, none has put forward this observation.

In Section \ref{subsec:multi} we introduce ``Multi-Lipschitz'' games, a generalization of Lipschitz games that allows player-specific Lipschitz values (the amount of influence the player has on others). We show that certain results of Lipschitz games extend to these, and the measure of interest is the sum of individual Lipschitz values (in a standard Lipschitz game, they are all equal). Theorem \ref{thm:multi} provides a query-efficient reduction from finding equilibria in Multi-Lipschitz games to finding equilibria in Lipschitz games. In particular, if there is a query-efficient approximation algorithm for the latter, there is one for the former as well.

Finally, Section \ref{subsec:det} provides a simpler proof of the result of \cite{HN18} showing exponential query lower-bounds on finding correlated equilibria with approximation constants better than $\frac{1}{2}$.  Theorem \ref{thm:det} provides a more general result for games with more than $2$ actions, and Corollary \ref{cor:det} extends this idea futher to apply to Lipschitz games.  While \cite{HN18} relies on a reduction from the $\AS$ problem, we explicitly describe a class of games with vastly different equilibria between which no algorithm making a subexponential number of queries can distinguish.  To any weak deterministic algorithm, these games look like pairs of players playing Matching Pennies against each other - however the equilibria are far from those of the Matching Pennies game.

For the sake of brevity, some technical details are omitted from this work, and can be found in the appendix.

\section{Preliminaries}\label{sec:prelim}

Throughout, we use the following notation.
\begin{itemize}
\item Boldface capital letters denote matrices, and boldface lowercase letters denote vectors.
\item The symbol $\veca$ is used to denote a pure action profile, and $\vecp$ is used when the strategy profiles may be mixed.  Furthermore, $\matX$ is used to denote \emph{correlated} strategies.
\item $[n]$ and $[m]$ denote the sets $\{1,\ldots,n\}$ of players and $\{1,\ldots,m\}$ of actions, respectively.  Furthermore, $i\in[n]$ will always refer to a player, and $j\in[m]$ will always refer to an action.
\item Whenever a query returns an approximate answer, the payoff vector $\tilde{\vecu}$ will be used to represent the approximation and $\vecu$ will represent the true value.
\end{itemize}

\subsection{The Game Model}\label{subsec:game}

We introduce standard concepts of strategy profiles, payoffs, regret, and equilibria for pure, mixed, and correlated strategies.

\begin{paragraph}{Types of strategy profile; notation:}
    \begin{itemize}
        \item A \emph{pure} action profile $\veca=(a_1,\ldots,a_n)\in[m]^n$ is an assignment of one action to each player. We use $\veca_{-i}=(a_1,\ldots,a_{i-1},a_{i+1},\ldots,a_n)\in[m]^{n-1}$ to denote the set of actions played by players in $[n]\setminus\{i\}$.
        \item A (possibly \emph{mixed}) strategy profile $\vecp=(\otherp_1,\ldots,\otherp_n)\in(\Delta[m])^n$ (where $\Delta(S)$ is the probability simplex over $S$) is a collection of $n$ independent probability distributions, each taken over the action set of a player, where $p_{ij}$ is the probability with which player $i$ plays action $j$.  The set of distributions for players in $[n]\setminus\{i\}$ is denoted $\vecp_{-i}=(\otherp_1,\ldots,\otherp_{i-1},\otherp_{i+1},\ldots,\otherp_n)$.  When $\vecp$ contains just $0$-$1$ values, $\vecp$ is equivalent to some action profile $\veca\in[m]^n$.
        
        Furthermore, when considering binary-action games with action set $\{1,2\}$, we instead describe strategy profiles by $\vecp=\left(p_1,\ldots,p_n\right)$, where $p_i$ is the probability that player $i$ plays action $1$.
        \item A \emph{correlated} strategy profile $\matX\in\Delta([m]^n)$ is a single joint probability distribution taken over the space of all pure action profiles $\veca$.
    \end{itemize}
\end{paragraph}

\begin{paragraph}{Notation for payoffs:}
    Given player $i$, action $j$, and pure action profile $\veca$,
    \begin{itemize}
        \item $u_i(j,\veca_{-i})$ is the payoff that player $i$ obtains for playing action $j$ when all other players play the actions given in $\veca_{-i}$.
        \item $u_i(\veca)=u_i(a_i,\veca_{-i})$ is the payoff that player $i$ obtains when all players play the actions given in $\veca$.
        \item Similarly for mixed-strategy profiles:\newline$u_i(j,\vecp_{-i})=\E_{\veca_{-i}\sim\vecp_{-i}}[u_i(j,\veca_{-i})]$ and $u_i(\vecp)=\E_{\veca\sim\vecp}[u_i(\veca)]$.
        \item For a given player $i\in[n]$, consider a deviation function $\phi:[m]\rightarrow[m]$.  Then, similarly, $u_i^{(\phi)}(\matX)=\E_{\veca\sim\matX}[u_i(\phi(a_i),\veca_{-i})]$ and $u_i(\matX)=\E_{\veca\sim\matX}[u_i(\veca)]$.  Furthermore, given an event $E$, $u_i(\matX\mid E)=\E_{\veca\sim\matX}[u_i(\veca)\mid E]$.
    \end{itemize}
\end{paragraph}

\begin{definition}[Regret]\label{def:reg}
    \begin{itemize}
        \item Given a player $i$ and a strategy profile $\vecp$, define the regret
        \[\reg_i(\vecp)=\max_{j\in[m]}u_i(j,\vecp_{-i})-u_i(\vecp)\]
        to be the difference between the payoffs of player $i$'s best response to $\vecp_{-i}$ and $i$'s strategy $p_i$.
        \item Given a player $i$ and a correlated strategy profile $\matX$, define
        \[\reg_i^{(\phi)}(\matX)=u_i^{(\phi)}(\matX)-u_i(\matX),\qquad\qquad\reg_i(\matX)=\max_{\phi:[m]\rightarrow[m]}\reg_i^{(\phi)}(\matX),\]
        the regret $\reg_i(\matX)$ being the difference between the payoffs of player $i$'s best deviation from $\matX$, and $\matX$.
    \end{itemize}
\end{definition}

\begin{definition}[Equilibria]
    \begin{itemize}
        \item An {\em $\epsilon$-approximate Nash equilibrium} ($\epsilon$-ANE) is a strategy profile $\vecp^*$ such that, for every player $i\in[n]$, $\reg_i(\vecp^*)\leq\epsilon$.
        \item An {\em $\epsilon$-well supported Nash equilibrium} ($\epsilon$-WSNE) is an $\epsilon$-ANE $\vecp^*$ for which every action $j$ in the support of $\otherp^*_i$ is an $\epsilon$-best response to $\vecp^*_{-i}$.  
        \item An {\em $\epsilon$-approximate pure Nash equilibrium} ($\epsilon$-PNE) is a pure action profile $\veca$ such that, for every player $i\in[n]$, $\reg_i(\veca)\leq\epsilon$.
        \item An {\em $\epsilon$-approximate correlated equilibrium} ($\epsilon$-ACE) is a correlated strategy profile $\matX^*$ such that, for every player $i\in[n]$, $\reg_i(\matX^*)\leq\epsilon$.
    \end{itemize}
\end{definition}

Note that any $\epsilon$-PNE is an $\epsilon$-WSNE, and any $\epsilon$-ANE constitutes an $\epsilon$-ACE.  Consequently, any algorithmic lower bounds on correlated equilibria also apply to Nash equilibria.

Finally, to end this section, we introduce the class of games that is the focus of this work.

\begin{definition}[Lipschitz Games]
    For any value $\lambda\in(0,1]$, a {\em $\lambda$-Lipschitz game} is a game in which a change in strategy of any given player can affect the payoffs of any other player by at most an additive $\lambda$, or for every player $i$ and pair of action profiles $\veca,\veca'$, $|u_i(\veca)-u_i(\veca')|\leq\lambda||\veca_{-i}-\veca'_{-i}||_1$.  Here we consider games in which all payoffs are in the range $[0,1]$ (in particular, note that any general game is, by definition, $1$-Lipschitz).
\end{definition}
The set of $n$-player, $m$-action, $\lambda$-Lipschitz games will be denoted $\G nm\lambda$.

\subsection{The Query Model}

This section introduces the model of queries we consider.

\begin{definition}[Queries]\label{def:query}
\begin{itemize}
    \item A {\em profile query} of a pure action profile $\veca$ of a game $G$, denoted $\mathcal{Q}^G(\veca)$, returns a vector $\vecu$ of payoffs $u_i(\veca)$ for each player $i\in[n]$.
    \item A {\em $\delta$-distribution query} of a strategy profile $\vecp$ of a game $G$, denoted $\mathcal{Q}^G_{\delta}(\vecp)$, returns a vector $\tilde{\vecu}$ of $n$ values such that $\left|\left|\tilde{\vecu}-\vecu\right|\right|_{\infty}\leq\delta$, where $\vecu$ is the players' expected utilities from $\vecp$.  We also define a {\em $(\delta,\gamma)$-distribution query} to be a $\delta$-distribution query of a strategy profile $\vecp$ in which every action $j$ in the support of $\otherp_i$ is allocated probability at least $\gamma$ for every player $i\in[n]$.
    \item The {\em (profile) query complexity} of an algorithm $A$ on input game $G$ is the number of calls $A$ makes to $\mathcal{Q}^G$.  The {\em $\delta$-distribution query complexity} of $A$ is the number of calls $A$ makes to $\mathcal{Q}^G_\delta$.
\end{itemize}
\end{definition}

Babichenko \cite{Bab16} points out that it is uninteresting to consider $0$-distribution queries, as any game in which every payoff is a multiple of $\frac{1}{M}$ for some $M\in\N$ can be completely learned by a single $0$-distribution query. On the other hand, additive approximations to the expected payoffs can be computed via sampling from $\vecp$. Indeed, for general binary-action games we have from \cite{GR16}:

\begin{theorem}[\cite{GR16}]\label{thm:GR16}
    Take $G\in\G n21,\eta>0$.  Any $(\delta,\gamma)$-distribution query of $G$ can be simulated with probability at least $1-\eta$ by
    \[\max\left\{\frac{1}{\gamma\delta^2}\log\left(\frac{8n}{\eta}\right),\frac{8}{\gamma}\log\left(\frac{4n}{\eta}\right)\right\}\]
    profile queries.
\end{theorem}

\begin{corollary}\label{cor:GR16}
    Take $G\in\G n21,\eta>0$.  Any $(\delta,\gamma)$-distribution query of $G$ can be simulated with probability at least $1-\eta$ by
    \[\frac{8}{\gamma^2\delta^2}\log^2\left(\frac{8n}{\eta}\right)\]
    profile queries.  Furthermore, any algorithm making $q$ $(\delta,\gamma)$-distribution queries of $G$ can be simulated with probability at least $1-\eta$ by
    \[\frac{8q}{\gamma^2\delta^2}\log^2\left(\frac{8nq}{\eta}\right)=\poly\left(n,\frac{1}{\gamma},\frac{1}{\delta},\log\frac{1}{\eta}\right)\cdot q\log q\]
    profile queries.
\end{corollary}

\begin{proof}
The first claim is a weaker but simpler version of the upper bound of Theorem \ref{thm:GR16}. The second claim follows from the first by a union bound.\qed
\end{proof}

\subsection{The Induced Population Game}

Finally, this section introduces a reduction utilized by \cite{AS13} in an alternative proof of Nash's Theorem, and by \cite{Bab13a} to upper bound the support size of $\epsilon$-ANEs.

\begin{definition}\label{def:gG}
    Given a game $G$ with payoff function $\vecu$, we define the {\em population game} induced by $G$, $G'=g_G(L)$ with payoff function $\vecu'$ as follows.  Every player $i$ is replaced by a population of $L$ players ($v^i_\ell$ for $\ell\in[L]$), each playing $G$ against the aggregate behavior of the other $n-1$ populations. More precisely,\\ $u'_{v^i_\ell}(\vecp')=u_i\left(p'_{v^i_\ell},\vecp_{-i}\right)$ where $p_{i'}=\frac{1}{L}\sum_{\ell=1}^Lp'_{v^{i'}_{\ell}}$ for all $i'\neq i$.
\end{definition}

Population games date back even to Nash's thesis \cite{N50}, in which he uses them to justify the consideration of mixed equilibria.  To date, the reduction to the induced population game has been focused on proofs of existence.  We show that the reduction can be made query-efficient: an equilibrium of $g_G(L)$ induces an equilibrium on $G$ \emph{which can be found with few additional queries}.  This technique is the foundation for the main results of this work.

\begin{lemma}\label{lem:Nfactor}[Appendix \ref{app:Nfactor}]
    Given an $n$-player, $m$-action game $G$ and a population game $G'=g_G(L)$ induced by $G$, if an $\epsilon$-PNE of $G'$ can be found by an algorithm making $q$ $(\delta,\gamma)$-distribution queries of $G'$, then an $\epsilon$-WSNE of $G$ can be found by an algorithm making $n\cdot m\cdot q$ $(\delta,\gamma/L)$-distribution queries of $G$.
\end{lemma}

\section{Results}\label{sec:results}

In this section, we present our three main results:

\begin{itemize}
    \item In Section \ref{subsec:pure}, Theorem \ref{thm:pure} shows a lower bound exponential in $\frac{n\lambda}{\epsilon}$ on the randomized query complexity of finding $\epsilon$-approximate \emph{pure} Nash equilibria of games in $\G n2\lambda$.
    \item In Section \ref{subsec:multi}, we generalize the concept of Lipschitz games.  Theorem \ref{thm:multi} provides a reduction from finding approximate equilibria in our new class of ``Multi-Lipschitz'' games to finding approximate equilibria of Lipschitz games.
    \item In Section \ref{subsec:det}, Theorem \ref{thm:det} and Proposition \ref{prop:det} provide a complete dichotomy of the query complexity of deterministic algorithms finding $\epsilon$-approximate correlated equilibria of $n$-player, $m$-action games.  Corollary \ref{cor:det} scales the lower bound to apply to Lipschitz games, and motivates the consideration of explicitly randomized algorithms for the above results.
\end{itemize}

These results also use the following simple lemma (which holds for all types of queries and equilibria mentioned in Section \ref{sec:prelim}).

\begin{lemma}\label{lem:scale}
    For any constants $\lambda'<\lambda\leq1,\epsilon>0$, there is a query-free reduction from finding $\epsilon$-approximate equilibria of games in $\G nm\lambda$ to finding $\frac{\lambda'}{\lambda}\epsilon$-approximate equilibria of games in $\G nm{\lambda'}$.
\end{lemma}

In other words, query complexity upper bounds hold as $\lambda$ and $\epsilon$ are scaled up together, and query complexity lower bounds hold as they are scaled down.  The proof is very simple - the reduction multiplies every payoff by $\frac{\lambda'}{\lambda}$ (making no additional queries) and outputs the result.  Note that the lemma does not hold for $\lambda'>\lambda$, as the reduction could introduce payoffs that are larger than $1$.

\subsection{Hardness of Approximate Pure Equilibria}\label{subsec:pure}

In this section we will rely heavily on the following result of Babichenko.

\begin{theorem}[\cite{Bab16}]\label{thm:Bab16}
    There is a constant $\epsilon_0>0$ such that, for any $\beta=2^{-o(n)}$, the randomized $\delta$-distribution query complexity of finding an $\epsilon_0$-WSNE of $n$-player binary-action games with probability at least $\beta$ is $\delta^22^{\Omega(n)}$.
\end{theorem}

For the remainder of this work, the symbol $\epsilon_0$ refers to this specific constant.  A simple application of Lemma \ref{lem:scale} yields

\begin{corollary}\label{cor:Bab16}
    There is a constant $\epsilon_0>0$ such that, for any $\beta=2^{-o(n)}$, the randomized $\delta$-distribution query complexity of finding an $\epsilon_0\lambda$-WSNE of games in $\G n2\lambda$ with probability at least $\beta$ is $\delta^22^{\Omega(n)}$.
\end{corollary}

We are now ready to state our main result -- an exponential lower bound on the randomized query complexity of finding $\epsilon$-PNEs of $\lambda$-Lipschitz games.

\begin{theorem}[Main Result]\label{thm:pure}
    There exists some constant $\epsilon_0$ such that, for any $n\in\N,\epsilon<\epsilon_0,\lambda\leq\frac{\epsilon}{\sqrt{8n\log4n}}$, while every game in $\G n2\lambda$ has an $\epsilon$-PNE, any randomized algorithm finding such equilibria with probability at least $\beta=1/\poly(n)$ must make $\lambda^22^{\Omega(n\lambda/\epsilon)}$ profile queries.
\end{theorem}

The proof follows by contradiction.  Assume such an algorithm $A$ exists making $\lambda^22^{o(n\lambda/\epsilon)}$ profile queries, convert it to an algorithm $B$ making $\lambda^22^{o(n\lambda/\epsilon)}$ $\delta$-distribution queries, then use Lemma \ref{lem:Nfactor} to derive an algorithm $C$ finding $\epsilon_0\lambda$-WSNE in $\lambda$-Lipschitz games contradicting the lower bound of Corollary \ref{cor:Bab16}.

\begin{proof}
Assume that some such algorithm $A$ exists finding $\epsilon$-PNEs of games in $\G n2\lambda$ making at most $\lambda^22^{o(n\lambda/\epsilon)}$ profile queries.  Consider any $\epsilon<\epsilon_0,\lambda'<\frac{\epsilon}{\sqrt{8n\log4n}}$, and define $\lambda=\frac{\epsilon}{\epsilon_0},L=\frac{\lambda}{\lambda'},N=Ln$.  We derive an algorithm $C$ (with an intermediate algorithm $B$) that contradicts Corollary \ref{cor:Bab16}.

\begin{enumerate}[$A$]
    \item Note that $A$ finds $\frac{\epsilon}{2}$-PNEs of games in $\G N2{\frac{3\lambda'}{2}}$ with probability at least $\beta$ making at most $\lambda^{\prime2}2^{o\left(N\lambda'/\epsilon\right)}$ profile queries ($\beta$ can be amplified to constant).
    \item Let $\delta=\frac{\epsilon_0\lambda'}{4}$.  For any game $G'\in\G N2{\lambda'}$, consider an algorithm making $\delta$-distribution queries of \emph{pure action profiles} of $G'$ (introducing the uncertainty without querying mixed strategies).

\begin{claim}[Appendix \ref{app:G''}]
	There is a game $G''\in\G N2{\frac{3\lambda'}{2}}$ that is consistent with all $\delta$-distribution queries (i.e. $\vecu''(\veca)=\tilde{\vecu}'(\veca)$ for all queried $\veca$) in which no payoff differs from $G'$ by more than an additive $\delta$.  Futhermore, any $\frac{\epsilon}{2}$-PNE of $G''$ is an $\epsilon$-PNE of $G'$.  Figure \ref{fig:dist} visually depicts this observation.
\end{claim}

\begin{figure}[t]
    \centering
    \includegraphics[scale=0.3]{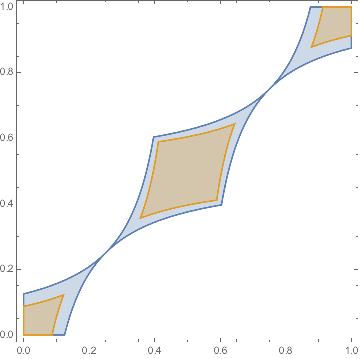}
    \caption{Taking $G'$ to be the Coordination Game for fixed values of $\epsilon$ and $\delta$, the blue region shows the set of $\epsilon$-approximate equilibria of $G'$ (the acceptable outputs of algorithm $B$) while the orange region shows the set of all $\frac{\epsilon}{2}$-approximate equilibria of any possible game $G''$ in which each payoff may be perturbed by at most $\delta$ (the possible outputs of algorithm $B$).}
    \label{fig:dist}
\end{figure}

So define the algorithm $B$ that takes input $G'$ and proceeds as though it is algorithm $A$ (but makes $\delta$-distribution queries instead).  By the claim above, after at most $\lambda^{\prime2}2^{o\left(N\lambda'/\epsilon\right)}$ queries, it has found an $\frac{\epsilon}{2}$-PNE of some $G''\in\G N2{\frac{3\lambda'}{2}}$ that it believes it has learned, which is also an $\epsilon$-PNE of $G'$.

    \item Consider any game $G\in\G n2\lambda$, and let $G'=g_G(L)$ be the population game induced by $G$.  There is an algorithm $C$ described by Lemma \ref{lem:Nfactor} that takes input $G$ and simulates $B$ on $G'$ (making $2n\cdot\lambda^{\prime2}2^{o\left(N\lambda'/\epsilon\right)}=\delta^22^{o\left(n\lambda/\epsilon\right)}$ $\delta$-distribution queries) and correctly outputs an $\epsilon$-WSNE (i.e. an $\epsilon_0\lambda$-WSNE) of $G$ with probability constant probability (so certainly $2^{-o(n)}$).
\end{enumerate}

The existence of algorithm $C$ directly contradicts the result of Corollary \ref{cor:Bab16}, proving that algorithm $A$ cannot exist.\qed
\end{proof}

\begin{remark}
    Note that, if we instead start the proof with the assumption of such an algorithm $B$, we can also show a $\delta^22^{o(n\lambda/\epsilon)}$ lower bound for the $\delta$-distribution query complexity of finding $\epsilon$-PNEs of $\lambda$-Lipschitz games.
\end{remark}

\subsection{Multi-Lipschitz Games}\label{subsec:multi}

In this section, we consider a generalization of Lipschitz games in which each player $i\in[n]$ has a ``player-specific'' Lipschitz value $\lambda_i$ in the sense that, if player $i$ changes actions, the payoffs of all other players are changed by at most $\lambda_i$.

\begin{definition}
    A $\Lambda$-Multi-Lipschitz game $G$ is an $n$-player, $m$-action game $G$ in which each player $i\in[n]$ is associated with a constant $\lambda_i\leq1$ such that $\sum_{i'=1}^n\lambda_{i'}=\Lambda$
    and, for any player $i'\neq i$ and action profiles $\veca^{(1)},\veca^{(2)}$ with $\veca^{(1)}_{-i}=\veca^{(2)}_{-i}$,
    $\left|u_{i'}\left(\veca^{(1)}\right)-u_{i'}\left(\veca^{(2)}\right)\right|\leq\lambda_i$.
    The class of such games is denoted $\GL\Lambda nm$, and for simplicity it is assumed that $\lambda_1\leq\ldots\leq\lambda_n$.
\end{definition}

The consideration of this generalized type of game allows real-world situations to be more accurately modeled.  Geopolitical circumstances, for example, naturally take the form of Multi-Lipschitz games, since individual countries have different limits on how much their actions can affect the rest of the world.  Financial markets present another instance of such games; they not only consist of individual traders who have little impact on each other, but also include a number of institutions that might each have a much greater impact on the market as a whole.  This consideration is further motivated by the recent GameStop frenzy; the institutions still wield immense power, but so do the aggregate actions of millions of individuals \cite{PL21}.

Notice that a $\lambda$-Lipschitz game is a $\Lambda$-Multi-Lipschitz game, for $\Lambda=n\lambda$. Any algorithm that finds $\epsilon$-ANEs of $\Lambda$-Multi-Lipschitz games is also applicable to $\Lambda/n$-Lipschitz games. Theorem \ref{thm:multi} shows a kind of converse for query complexity, reducing from finding $\epsilon$-ANE of $\Lambda$-Multi-Lipschitz games to finding $\epsilon$-ANE of $\lambda$-Lipschitz games, for $\lambda$ a constant multiple of $\Lambda/n$.

\begin{theorem}\label{thm:multi}
    There is a reduction from computing $\epsilon$-ANEs of games in $\GL\Lambda n2$ with probability at least $1-\eta$ to computing $\frac{\epsilon}{2}$-ANEs of games in $\G{2n}2{\frac{3\Lambda}{2n}}$ with probability at least $1-\frac{\eta}{2}$ introducing at most a multiplicative $\poly(n,\frac{1}{\epsilon},\log\frac{1}{\eta})$ query blowup.
\end{theorem}

As we now consider $\epsilon$-ANEs, existence is no longer a question: such equilibria are \emph{always} guaranteed to exist by Nash's Theorem \cite{N51}.  This proof will also utilize a more general population game $G'=g_G(L_1,\ldots,L_n)$ in which player $i$ is replaced by a population of size $L_i$ (where the $L_i$ may differ from each other), and the queries in Lemma \ref{lem:Nfactor} become $(\delta,\min_{i\in[n]}\{\gamma/L_i\})$-distribution queries (this will now be relevant, as we need to apply Corollary \ref{cor:GR16}).  Otherwise, the proof follows along the same lines as that of Theorem \ref{thm:pure}.

\begin{proof}
    Consider some $\epsilon>0$ and a game $G\in\GL\Lambda n2$ (WLOG take $\lambda_1\leq\ldots\leq\lambda_n)$.  First, if $\Lambda<\frac{\epsilon}{n}$, finding an $\epsilon$-ANE is trivial (each player can play their best-response to the uniform mixed strategy, found in $2n$ queries).  So assume $\Lambda\geq\frac{\epsilon}{n}$.  Define $L_i=\max\{\frac{n\lambda_i}{\Lambda},1\}$ and, taking $i'=\max_{i\in[n]}\{i:L_i=1\}$, note that
    \[\sum_{i=1}^nL_i=\sum_{i=1}^{i'}1+\sum_{i=i'+1}^n\frac{n\lambda_i}{\Lambda}=\sum_{i=1}^{i'}1+\frac{n}{\Lambda}\sum_{i=i'+1}^n\lambda_i\leq\sum_{i=1}^{i'}1+\frac{n}{\Lambda}\Lambda\leq2n.\]
    Thus the population game $G'=g_G(L_1,\ldots,L_n)\in\G{2n}2{\frac{\Lambda}{n}}$.
    \begin{enumerate}[$A$]
        \item Consider an algorithm $A$ that finds $\frac{\epsilon}{2}$-ANEs of games in $\G{2n}2{\frac{3\Lambda}{2n}}$, with probability at least $1-\frac{\eta}{2}$ making $q$ profile queries.
        \item Taking $\delta=\frac{\epsilon^2}{4n^2}<\frac{\epsilon\Lambda}{4n}$, the algorithm $B$ from the proof of Theorem \ref{thm:pure} that simulates $A$ but makes $(\delta,1)$-distribution queries finds an $\epsilon$-ANE of $G'$ (the proof in Appendix \ref{app:G''} also holds for these parameters with this choice of $\delta$).
        \item By Lemma \ref{lem:Nfactor}, there is an algorithm $C$ on input $G\in\GL\Lambda n2$ that simulates $B$ (replacing each $(\delta,1)$-distribution query of $G'$ with $2n$ $(\delta,\frac{1}{n})$-distribution queries of $G$ since $\frac{1}{L_n}\geq\frac{1}{n}$) finding an $\epsilon$-ANE with probability at least $1-\eta$.
    \end{enumerate}
    Applying Corollary \ref{cor:GR16} (using $\delta=\frac{\epsilon^2}{4n^2},\gamma=\frac{1}{n}$) to create a profile-query algorithm from $C$ completes the proof.\qed
\end{proof}

As an example application of Theorem \ref{thm:multi}, an algorithm of \cite{GCW19} finds $\left(\frac{1}{8}+\alpha\right)$-approximate Nash equilibria of games in $\G n2{\frac{1}{n}}$; Theorem \ref{thm:GCW19} states that result in detail, and Corollary \ref{cor:GCW} extends it to Multi-Lipschitz games.

\begin{theorem}[\cite{GCW19}]\label{thm:GCW19}
  Given constants $\alpha,\eta>0$, there is a randomized algorithm that, with probability at least $1-\eta$, finds $\left(\frac{1}{8}+\alpha\right)$-approximate Nash equilibria of games in $\G n2{\frac{1}{n}}$ making $O\left(\frac{1}{\alpha^4}\log\left(\frac{n}{\alpha\eta}\right)\right)$ profile queries.  
\end{theorem}

We now have some ability to apply this to Multi-Lipschitz games; if $1\leq\Lambda<4$ we can improve upon the trivial $\frac{1}{2}$-approximate equilibrium of Proposition \ref{prop:det}.

\begin{corollary}\label{cor:GCW}
    For $\alpha,\eta>0,\Lambda\geq1,\epsilon\geq\frac{\Lambda}{8}+\alpha$, there is an algorithm finding $\epsilon$-ANEs of games in $\GL\Lambda n2$ with probability at least $1-\eta$ making at most $\poly(n,\frac{1}{\alpha},\log\frac{1}{\eta})$
    profile queries.
\end{corollary}

\begin{remark}
    This is actually a slight improvement over just combining Theorems \ref{thm:multi} and \ref{thm:GCW19}, since the choice of $\delta$ can be made slightly smaller to shrink $\alpha$ as necessary.
\end{remark}

\subsection{A Deterministic Lower Bound}\label{subsec:det}

We complete this work by generalizing the following result of Hart and Nisan.

\begin{theorem}[\cite{HN18}]\label{thm:HN18}
    For any $\epsilon<\frac{1}{2}$, the deterministic profile query complexity of finding $\epsilon$-ACEs of $n$-player games is $2^{\Omega(n)}$.
\end{theorem}

While the proof of Theorem \ref{thm:HN18} utilizes a reduction from $\AS$, we employ a more streamlined approach, presenting an explicit family of ``hard'' games that allows us to uncover the optimal value of $\epsilon$ as a function of the number of actions:

\begin{theorem}\label{thm:det}
    Given some $m\in\N$, for any $\epsilon<\frac{m-1}{m}$, the deterministic profile query complexity of finding $\epsilon$-ACEs of $n$-player, $m$-action games is $2^{\Omega(n)}$.
\end{theorem}

Furthermore, this value of $\epsilon$ cannot be improved:

\begin{proposition}\label{prop:det}
    Given some $n,m\in\N$, for any $\epsilon\geq\frac{m-1}{m}$, an $\epsilon$-ANE of an $n$-player, $m$-action game can be found making no profile queries.
\end{proposition}

The upper bound of Proposition \ref{prop:det} can be met if every player plays the uniform mixed strategy over their actions.  Finally, we can apply Lemma \ref{lem:scale} to scale Theorem \ref{thm:det} and obtain our intended result:

\begin{corollary}\label{cor:det}
    Given some $m\in\N,\lambda\in(0,1]$, for any $\epsilon<\frac{m-1}{m}\lambda$, the deterministic profile query complexity of finding $\epsilon$-ACEs of $n$-player, $m$-action, $\lambda$-Lipschitz games is $2^{\Omega(n)}$.
\end{corollary}

In order to prove these results, we introduce a family of games $\{G_{k,m}\}$.  For any $k,m\in\N$, $G_{k,m}$ is a $2k$-player, $m$-action generalization of $k$ Matching Pennies games in which every odd player $i$ wants to match the even player $i+1$ and every even player $i+1$ wants to mismatch with the odd player $i$.

\begin{definition}\label{def:gk}
    Define $G_{1,2}$ to be the generalized Matching Pennies game, as described in Figure \ref{fig:MP}(a).  Define the generalization $G_{k,m}$ to be the $2k$-player $m$-action game such that, for any $i\in[k]$, player $2i-1$ has a payoff $1$ for matching player $2i$ and $0$ otherwise (and vice versa for player $2i$) ignoring all other players.
\begin{figure}[t]
    \centering
    \begin{subfigure}[b]{0.45\textwidth}
        \centering
        \begin{tabular}{cc|c|c|}
            & \multicolumn{1}{c}{} & \multicolumn{2}{c}{Player $2$} \\
            & \multicolumn{1}{c}{} & \multicolumn{1}{c}{$1$}  & \multicolumn{1}{c}{$2$} \\\cline{3-4}
            \multirow{2}*{Player $1$} & $1$ & $1,0$ & $0,1$ \\\cline{3-4}
            & $2$ & $0,1$ & $1,0$ \\\cline{3-4}
        \end{tabular}
        \caption{The payoff matrix of $G_{1,2}$, the Matching Pennies game.}
    \end{subfigure}
    \begin{subfigure}[b]{0.54\textwidth}
        \centering
        \begin{tabular}{cc|c|c|c|}
            & \multicolumn{1}{c}{} & \multicolumn{3}{c}{Player 2} \\
            & \multicolumn{1}{c}{} & \multicolumn{1}{c}{$1$}  & \multicolumn{1}{c}{$2$}  &\multicolumn{1}{c}{$3$} \\\cline{3-5}
            \multirow{3}*{Player 1} & $1$ & $1,0$ & $0,1$ & $0,1$ \\\cline{3-5}
            & $2$ & $0,1$ & $1,0$ & $0,1$ \\\cline{3-5}
            & $3$ & $0,1$ & $0,1$ & $1,0$ \\\cline{3-5}
        \end{tabular}
        \caption{The payoff matrix of $G_{1,3}$, the generalized Matching Pennies game.}
    \end{subfigure}
    \caption{The payoff matrices of $G_{1,2}$ and $G_{1,3}$.}
    \label{fig:MP}
\end{figure}
\end{definition}

The critical property of the generalized Matching Pennies game is that we can bound the probability that any given action profile is played in any $\epsilon$-ACE of $G_{k,m}$.  If too much probability is jointly placed on matching actions, player $2$ will have high regret.  Conversely, if too much probability is jointly placed on mismatched actions, player $1$ will have high regret.

\begin{lemma}\label{lem:game}
    For any $k,m\in\N,\alpha>0$, take $\epsilon=\frac{m-1}{m}-\alpha$.  In any $\epsilon$-ACE $\matX^*$ of $G_{k,m}$, every action profile $\veca'\in[m]^n$ satisfies $\Pr_{\veca\sim\matX^*}(\veca=\veca')<\rho^{\frac{n}{2}}$ where
    \[\rho=\frac{(2-\alpha)m-1}{2m}.\]
\end{lemma}

This phenomenon can be seen in Figure \ref{fig:cregion}.

\begin{figure}[t]
    \centering
    \begin{subfigure}[b]{0.3\textwidth}
        \centering
        \includegraphics[scale=0.3]{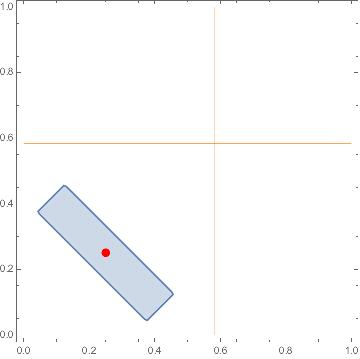}
        \caption{$\alpha=\frac{1}{3}$}
    \end{subfigure}
    \begin{subfigure}[b]{0.3\textwidth}
        \centering
        \includegraphics[scale=0.3]{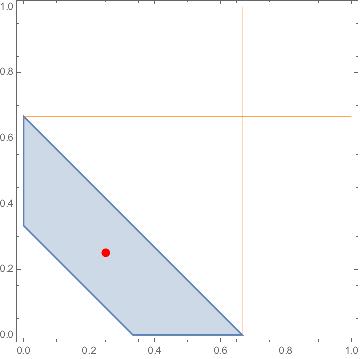}
        \caption{$\alpha=\frac{1}{6}$}
    \end{subfigure}
    \begin{subfigure}[b]{0.3\textwidth}
        \centering
        \includegraphics[scale=0.3]{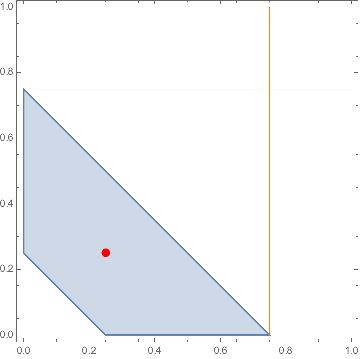}
        \caption{$\alpha\rightarrow0$}
    \end{subfigure}
    \caption{The region of possible values for $(\Pr_{\veca\sim\matX^*}(a_1=1,a_2=1),\Pr_{\veca\sim\matX^*}(a_1=2,a_2=2)$ in any $\left(\frac{1}{2}-\alpha\right)$-approximate correlated equilibrium of $G_{k,2}$.  The only exact correlated equilibrium is shown by the red point, and the corresponding values of $\rho$ are displayed as the orange lines.}
    \label{fig:cregion}
\end{figure}

\begin{proof}
    Define $n=2k$ to be the number of players in $G_{k,m}$, and consider some $\epsilon$-ACE $\matX^*$.  Now WLOG consider players $1$ and $2$ and assume, for the sake of contradiction, that there exist some actions $j_1,j_2$ such that $\Pr_{\veca\sim\matX^*}(a_1=j_1,a_2=j_2)>\rho$.  We will need to consider the two cases $j_1=j_2$ and $j_1\neq j_2$.
    \item\paragraph{Matching Actions} In this case, WLOG assume $j_1=j_2=1$.  We show that player $2$ can improve her payoff by more than $\epsilon$.  Under $\matX^*$, with probability $>\rho$, any random realization $\veca\sim\matX^*$ will yield player $2$ a payoff of $0$.  In other words, $u_2(\matX^*)<1-\rho$.  Furthermore, considering the marginal distribution over player $1$'s action, we are guaranteed that
    \[\sum_{j=2}^m\Pr_{\veca\sim\matX^*}(a_1=j)<1-\rho\]
    so there must exist some action (WLOG action $2$) for which $\Pr_{\veca\sim\matX^*}(a_1=2)<\frac{1-\rho}{m-1}$.  As such, define $\phi(j)=2$.  Then $u_2^{(\phi)}(\matX^*)>1-\frac{1-\rho}{m-1}$, so
    \[\reg_2^{(\phi)}(\matX^*)>\underbrace{\left(1-\frac{1-\rho}{m-1}\right)}_{u_2^{(\phi)}(\matX^*)}-\underbrace{(1-\rho)}_{u_2(\matX^*)}=\frac{\rho m-1}{m-1}\geq\frac{m-1}{m}-\alpha\]
    for all $m\geq2$.  This contradicts our assumption of an $\epsilon$-ACE.
    \item\paragraph{Mismatched Actions}
    In this case, WLOG assume $j_1=1,j_2=2$.  The situation is simpler, taking $\phi(j)=2$, $u_1(\matX^*)<1-\rho$ and $u_1^{(\phi)}(\matX^*)>\rho$, so $\reg_1^{(\phi)}(\matX^*)>\frac{m-1}{m}-\alpha$.  This too contradicts our assumption of an $\epsilon$-ACE, and thus completes the proof of Lemma \ref{lem:game}.\qed
\end{proof}

We can now prove Theorem \ref{thm:det}.  The general idea is that, should an efficient algorithm exist, because any equilibrium of $G_{k,m}$ must have large support by Lemma \ref{lem:game}, there is significant probability assigned to action profiles that are not queried by the algorithm.  We show there is a game that the algorithm cannot distinguish from $G_{k,m}$ that shares no approximate equillibria with $G_{k,m}$.

\begin{proof}[Theorem \ref{thm:det}]
    Consider any $\alpha>0$ and let $\epsilon=\frac{m-1}{m}-\alpha$.  Taking $\rho$ as in the statement of Lemma \ref{lem:game}, assume there exists some deterministic algorithm $A$ that takes an $n$-player, $m$-action game $G$ as input and finds an $\epsilon$-ACE of $G$ querying the payoffs of $q<\frac{\alpha}{2}\rho^{-\frac{n}{2}}$ action profiles.  Fix some $k\in\N$ and consider input $G_{k,m}$ as defined in Definition \ref{def:gk}.  Then $\matX^*=A\left(G_{k,m}\right)$ is an $\epsilon$-ACE of $G_{k,m}$.  Note that, for some $j$, $\Pr_{\veca\sim\matX^*}(a_1=j)\leq\frac{1}{m}$ (WLOG assume $j=1$).
    
    Now define the perturbation $G'_{k,m}$ of $G_{k,m}$ with payoffs defined to be equal to $G_{k,m}$ for every action profile queried by $A$, $1$ for every remaining action profile in which player $1$ plays action $1$ (chosen because it is assigned low probability by $\matX^*$ by assumption), and $0$ otherwise.  Note that, by definition, $A$ cannot distinguish between $G_{k,m}$ and $G'_{k,m}$, so $A(G'_{k,m})=\matX^*$.
    
    Taking the function $\phi(j)=1$, the quantity we need to bound is $\reg_i^{\prime(\phi)}(\matX^*)\geq u_i^{\prime(\phi)}(\matX^*)-u'_i(\matX^*)$.
    We must bound the components of this expression as follows:
\begin{claim}[Appendix \ref{app:Gkm}]
$u^{\prime(\phi)}_1(\matX^*)>(1-q\rho^{\frac{n}{2}})$ and $u'_1(\matX^*)<\left(\frac{1}{m}+q\rho^{\frac{n}{2}}\right)$.
\end{claim}
    Using the claim and once again recalling the assumption that $q<\frac{\alpha}{2}\rho^{-\frac{n}{2}}$, we see
    \[\reg^{\prime(\phi)}_1(\matX^*)>\left(1-\frac{1}{m}-2q\rho^{\frac{n}{2}}\right)=\frac{m-1}{m}-\alpha=\epsilon.\]
    So $\matX^*$ cannot actually be an $\epsilon$-ACE of $G_{k,m}$. This completes the proof of Theorem \ref{thm:det}.\qed
\end{proof}

\section{Further Directions}

An important additional question is the query complexity of finding $\epsilon$-PNEs of $n$-player, $\lambda$-Lipschitz games in which $\epsilon=\Omega(n\lambda)$.  Theorem \ref{thm:pure} says nothing in this parameter range, yet Theorem \ref{thm:GCW19} provide a logarithmic upper bound in this regime.  The tightness of this bound is of continuing interest.  Furthermore, the query- and computationally-efficient reduction discussed in Lemma \ref{lem:Nfactor} provides a hopeful avenue for further results bounding the query, and computational, complexities of finding equilibria in many other classes of games.

\paragraph{Acknowledgements}{We thank Francisco Marmalejo Coss\'{i}o and Rahul Santhanam for their support in the development of this work, and the reviewers of an earlier version for helpful comments.  Matthew Katzman was supported by an Oxford-DeepMind Studentship for this work.}

\bibliographystyle{splncs04}
\bibliography{mybibliography}

\newpage

\appendix

\section{Proof of Lemma \ref{lem:Nfactor}: A Query-Efficient Reduction}\label{app:Nfactor}

Lemma \ref{lem:Nfactor} is reproduced below:

\begin{lemma}[Lemma \ref{lem:Nfactor} Restated]
    Given an $n$-player, $m$-action game $G$ and a population game $G'=g_G(L)$ induced by $G$, if an $\epsilon$-PNE of $G'$ can be found by an algorithm making $q$ $(\delta,\gamma)$-distribution queries of $G'$, then an $\epsilon$-WSNE of $G$ can be found by an algorithm making $n\cdot m\cdot q$ $(\delta,\gamma/L)$-distribution queries of $G$.
\end{lemma}

We first establish the following preliminary claim:

\begin{claim}
    Any $\delta$-distribution query of $G'=g_G(L)$ can be simulated by $Ln$ $\delta$-distribution queries of $G$.
\end{claim}

At a high level, this is because we will make a single query of the original game for each player $v^{(i)}_\ell$ in the population game in which we assume that player $i$ plays the strategy of player $v^{(i)}_\ell$ and the remaining $n-1$ players play the aggregate behavior of their entire populations.

\begin{proof}
    Consider any $\delta$-distribution query of the $Ln$-player game $G'$.  Such a query takes a mixed strategy profile $\vecp'$ and returns a $\delta$-approximate payoff for each of the $Ln$ players.  If we consider the aggregate strategy profile $\vecp$ directly, the query $\mathcal{Q}^G_\delta(\vecp)$ will return the payoff to each of the $n$ players in $G$.  However, we cannot derive the payoff to a player $v^{(i)}_\ell$ in $G'$ from the average payoff to each such player.  Instead, for every $i\in[n],\ell\in[L]$, consider the strategy profile $\vecp^{(i)}_\ell$ replacing player $i$ with player $v^{(i)}_\ell$ and replacing the remaining $n-1$ players by the aggregate behavior of their populations.  Then, a query $\mathcal{Q}^G_\delta(\vecp^{(i)}_\ell)$ will return a vector of $n$ payoffs, the $i^{\textnormal{th}}$ element of which will correctly be a $\delta$-approximation of the payoff to player $v^{(i)}_\ell$ when $\vecp'$ is played in game $G'$.\qed
\end{proof}

Now, to complete the proof of the Lemma, we will improve the blowup factor from $Ln$ to $mn$.  In essence, this is because we do not actually need to make a query for every single player in each population.  Instead, if we query every pure strategy for player $i$ against the aggregate behavior of the other $n-1$ populations, we can then calculate the payoff to \emph{any} mixed strategy of a player $v^{(i)}_\ell$ in population $i$.

\begin{proof}[Lemma \ref{lem:Nfactor}]
    Consider any $\delta$-distribution query of the $Ln$-player game $G'$.  In the proof above, this query is simulated by a single query for each player in $G'$.  Instead, we can simulate it by $m$ queries for every player in $G$.  Instead of queries $\vecp^{(i)}_\ell$ for every $i\in[n],\ell\in[L]$, consider queries $\vecp^{(i)}_j$ for every $i\in[n],j\in[m]$ in which player $i$ plays action $j$ and the remaining $n-1$ players still play the aggregate behavior of their populations.  From these $n\cdot m$ queries, one can derive the payoffs from the $Ln$ queries $\vecp^{(i)}_\ell$ above by directly calculating the expected payoff to the mixed strategy played by player $v^{(i)}_\ell$ given the payoffs for each pure action.  So each $\delta$-distribution query of $G'$ can be simulated by $n\cdot m$ $\delta$-distribution queries of $G$.  Furthermore, if the query was a $(\delta,\gamma)$-distribution query of $G'$, the aggregate behavior of $L$ players guarantees that the new queries are all $(\delta,\gamma/L)$-distribution queries.  Finally, any $\epsilon$-PNE $\veca^{\prime*}$ of $G'$ must also induce an $\epsilon$-WSNE $\vecp^*$ of $G$ in which each player of $G$ plays the aggregate strategy of their population.  If $\vecp^*$ were not an $\epsilon$-WSNE, then some player in $G'$ must suffer regret $>\epsilon$ in $\veca^{\prime*}$.  This completes the proof of Lemma \ref{lem:Nfactor}.\qed
\end{proof}

\section{The Existence of Game $G''$}\label{app:G''}

Here we provide a proof of the Claim in the proof of Theorem \ref{thm:pure}.  Such a Claim is necessary to ensure that the introduction of the $\delta$ error in the queries of $G'$ still results in payoffs which, if interpreted as ground truth, $A$ would could still use to correctly find an equilibrium of $G'$.  Assuming an algorithm $B$ making $\delta$-distribution queries of the \emph{pure action profiles} of a game $G'\in\G N2{\lambda'}$ where $\delta=\frac{\epsilon_0\lambda'}{4}$, we reproduce the claim below:

\begin{claim}
	There is a game $G''$ consistent with all $\delta$-distribution queries (i.e. $\vecu''(\veca)=\tilde{\vecu}'(\veca)$ for all queried $\veca$) satisfying:
	\begin{enumerate}[(1)]
	    \item For any action profile $\veca\in\{0,1\}^n$, $||\tilde{\vecu}'(\veca)-\vecu'(\veca)||_\infty\leq\delta$.
	    \item $G''\in\G N2{\frac{3\lambda'}{2}}$.
	    \item Any $\frac{\epsilon}{2}$-PNE of $G''$ is an $\epsilon$-PNE of $G'$.
	\end{enumerate}
\end{claim}

\begin{proof}
    Consider the incomplete payoff function $\tilde{\vecu}'$ learned by $\delta$-distribution queries of $G'$.  Consider the game $G''$ agreeing with $\tilde{\vecu}$ on all queried actions and agreeing with $G'$ on all others:
    \[\vecu''(\veca)=\begin{cases}\tilde{\vecu}'(\veca)&\veca\textnormal{ was queried}\\\vecu'(\veca)&\textnormal{otherwise}\end{cases}\]
    Then for every action profile $\veca\in\{0,1\}^n$, $\left|\left|\tilde{\vecu}'(\veca)-\vecu'(\veca)\right|\right|_{\infty}\leq\delta$.  This proves point (1).
    
    Now, because $\epsilon_0<1$, $\delta<\frac{\lambda}{4}$.  Furthermore, consider any player $i\in[n]$ and any action profiles $\veca^{(1)},\veca^{(2)}\in\{0,1\}^n$ such that $\left|\left|\veca_{-i}^{(1)}-\veca_{-i}^{(2)}\right|\right|_1=c>1$.
    Then, since $G'$ is $\lambda'$-Lipschitz,
    \begin{align*}
        \left|\tilde{u}'_i\left(\veca^{(1)}\right)-\tilde{u}'_i\left(\veca^{(2)}\right)\right|&\leq\left|u'_i\left(\veca^{(1)}\right)-u'_i\left(\veca^{(2)}\right)\right|+2\delta\\
        &\leq c\lambda+2\delta\\
        &\leq c\left(\lambda+2\delta\right)\\
        &<c\frac{3\lambda}{2}
    \end{align*}
    This proves point (2).

    Finally, consider any $\frac{\epsilon}{2}$-approximate pure Nash equilibrium $\veca^*$ of $G'$, and note that, for any player $i\in[n]$,
    \[u'_i(\veca^*)\geq\tilde{u}'_i(\veca^*)-\delta,\qquad\qquad u'_i(j,\veca^*_{-i})\leq\tilde{u}'_i(j,\veca^*_{-i})+\delta\]
    for any $j\in\{0,1\}$.  Combining these two inequalities, and keeping in mind that $\lambda'<\frac{\epsilon}{\epsilon_0}$,
    \begin{align*}
        \reg'_i(\veca^*)&\leq\frac{\epsilon}{2}+2\delta\\
        &=\frac{\epsilon}{2}+\frac{\epsilon_0\lambda}{2}\\
        &<\frac{\epsilon}{2}+\frac{\epsilon}{2}\\
        &=\epsilon
    \end{align*}
    proving point (3) and the Claim.\qed
\end{proof}

\begin{remark}
    Considering instead games $G'\in\G{2n}2{\frac{\Lambda}{n}}$ and $G''\in\G{2n}2{\frac{3\Lambda}{2n}}$, the proof still goes through replacing $\delta=\frac{\epsilon\Lambda}{4n}$.  This instead yields the result claimed in the proof of Theorem \ref{thm:multi}.
\end{remark}

\section{Bounding the Payoffs in $G'_{k,m}$}\label{app:Gkm}

Here we provide a proof of the Claim in the proof of Theorem \ref{thm:det}.  As a reminder, we have assumed that a deterministic algorithm $A$ making $q=\frac{\alpha}{2}\rho^{-\frac{n}{2}}$ profile queries exists finding $\epsilon$-ACEs of $n$-player, $m$-action games, and are considering here the perturbation $G'_{k,m}$ of $G_{k,m}$ as defined in the proof.  Let $\matX^*=A(G_{k,m})$, an $\epsilon$-ACE of $G_{k,m}$, and WLOG assume $\Pr_{\veca\sim\matX^*}(a_1=1)\leq\frac{1}{m}$.  Define the sets $Q_i^{(j)}$ to consist of the pure action profiles queried by $A$ in which player $i$ plays action $j$ (as a basic bound, note that none of these sets can contain more than $q$ elements).  The perturbation $G'_{k,m}$ is the game with payoffs that agree with $G_{k,m}$ on all queried action profiles of $A$, and otherwise provide player $1$ a payoff of $1$ for playing action $1$ and a payoff of $0$ otherwise.  Because $A$ cannot distinguish between $G_{k,m}$ and $G'_{k,m}$, $\matX^*=A(G'_{k,m})$.  Taking $\phi(j)=1$, the claim is reproduced below in two parts.

\begin{claim}
    $u^{\prime(\phi)}_1(\matX^*)>(1-q\rho^{\frac{n}{2}})$.
\end{claim}

\begin{proof}
    We first separate the payoffs into those which were queried by $A$ and those which were not (i.e. those that differ between $G_{k,m}$ and $G'_{k,m}$ and those that do not):
    \[u^{\prime(\phi)}_1(\matX^*)=\sum_{\veca'\in Q_1^{(1)}}\Pr_{\veca\sim\matX^*}(\veca=\veca')u_1(1,\veca_{-1})+\sum_{\veca'\not\in Q_1^{(1)}}\Pr_{\veca\sim\matX^*}(\veca=\veca').\]
    Unfortunately, the best lower bound we can claim for the first sum is $0$.  However, the second sum is more friendly.
    \[u{\prime(\phi)}_1(\matX^*)\geq\sum_{\veca'\not\in Q_1^{(1)}}\Pr_{\veca\sim\matX^*}(\veca=\veca')\lambda=\left(1-\sum_{\veca'\in Q_1^{(1)}}\Pr_{\veca\sim\matX^*}(\veca=\veca')\right).\]
    Combining our assumption on the size of $Q_1^{(1)}$ with the result of Lemma \ref{lem:game} (which provides an upper bound on the value of $\Pr_{\veca\sim\matX^*}(\veca=\veca')$), this yields
    \[u^{\prime(\phi)}_1(\matX^*)>\left(1-q\rho^{\frac{n}{2}}\right),\]
    proving the Claim.\qed
\end{proof}

\begin{claim}
    $u'_1(\matX^*)<\left(\frac{1}{m}+q\rho^{\frac{n}{2}}\right)$.
\end{claim}

\begin{proof}
    It is slightly more involved to obtain an upper bound on the value of $u'_1(\matX^*)$.  We first note the following.
    \[u'_1(\matX^*)=\Pr_{\veca\sim\matX^*}(a_1=1)u'_1(\matX^*\mid x_1=1)+\Pr_{\veca\sim\matX^*}(a_1\neq1)u'_1(\matX^*\mid x_1\neq1).\]
    By assumption, the first expression is upper bounded at $\frac{1}{m}$.  Our focus, therefore, must lie on the second expression.  This can be broken up further by the Law of Total probability:
    \begin{align*}
        &\quad\Pr_{\veca\sim\matX^*}(a_1\neq1)u'_1(\matX^*\mid x_1\neq1)\\
        &=\Pr_{\veca\in\matX^*}(a_1\neq1,\veca\in Q_1^{(a_1)})u'_1(\matX^*\mid a_1\neq1,\veca\in Q_1^{(a_1)})\\
        &\qquad+\Pr_{\veca\in\matX^*}(a_1\neq1,\veca\not\in Q_1^{(a_1)})u'_1(\matX^*\mid a_1\neq1,\veca\not\in Q_1^{(a_1)})
    \end{align*}
    Since $u'_1(\matX^*\mid a_1\neq1,\veca\not\in Q_1^{(a_1)})=0$,
    we need only consider the case when $\veca\in Q_1^{(a_1)}$.  While we cannot bound the payoff we can, fortunately, bound the probability
    \[\Pr_{\veca\sim\matX^*}(a_1\neq1,\veca\in Q_1^{(a_1)})\leq q\rho^{\frac{n}{2}}.\]
    This bounds the entire payoff at
    \[u'_1(\matX^*)\leq\left(\frac{1}{m}+q\rho^{\frac{n}{2}}\right)\]
    (note that we can do even better, since we are double-counting some, but this is sufficient).\qed
\end{proof}

\end{document}